\documentclass[10pt, conference, compsocconf]{IEEEtran}
\usepackage{amssymb,amsmath,amsthm}
\usepackage{multirow,booktabs}
\usepackage{pgf}
\usepackage{xspace}
\usepackage{cite,url}

\theoremstyle{plain}

\newtheorem{proposition}{Proposition}
\newtheorem{lemma}{Lemma}
\newtheorem{corollary}{Corollary}

\theoremstyle{definition}
\newtheorem{definition}{Definition}

\theoremstyle{remark}
\newtheorem{remark}{Remark}

\newcommand{\renyi}{R\'enyi\xspace}
\newcommand{\supp}{\operatorname{supp}}
\newcommand{\E}{\operatorname{E}}
\newcommand{\D}{\mathcal{D}}
\newcommand{\R}{\mathcal{R}}

\newcommand{\eps}{\ensuremath{\epsilon}}
\newcommand{\alphaeps}{\ensuremath{(\alpha, \eps)}}
\newcommand{\epsdelta}{\ensuremath{(\eps,\delta)}\xspace}
\newcommand{\dx}{\,\mathrm{d}x}
\newcommand{\dy}{\,\mathrm{d}y}
\newcommand{\LaplaceM}{\operatorname{\textbf{L}}}
\newcommand{\GaussianM}{\operatorname{\textbf{G}}}
\newcommand{\RRM}{\operatorname{\textnormal{\textbf{RR}}}}
\newcommand{\leader}[1]{\medskip\noindent\textsc{#1}}

\begin{document}

\author{\IEEEauthorblockN{Ilya Mironov}
\IEEEauthorblockA{Google Brain}}

\title{\renyi\ Differential Privacy}

\maketitle

\begin{abstract}
We propose a natural relaxation of differential privacy based on the \renyi divergence. Closely related notions have appeared in several recent papers that analyzed composition of differentially private mechanisms. We argue that the useful analytical tool can be used as a privacy definition, compactly and accurately representing guarantees on the tails of the privacy loss. 

We demonstrate that the new definition shares many important properties with the standard definition of differential privacy, while additionally allowing tighter analysis of composite heterogeneous mechanisms.
\end{abstract}

\section{Introduction}

Differential privacy, introduced by Dwork et al.~\cite{DMNS06}, has been embraced by multiple research communities as a commonly accepted notion of privacy for algorithms on statistical databases. As applications of differential privacy begin to emerge, practical concerns of \emph{tracking} and \emph{communicating} privacy guarantees are coming to the fore.

Informally, differential privacy bounds a shift in the output distribution of a randomized algorithm that can be induced by a small change in its input. The standard definition of \eps-differential privacy puts a multiplicative upper bound on the worst-case change in the distribution's density.

Several relaxations of differential privacy explored other measures of closeness between two distributions. The most common such relaxation, the \epsdelta definition, has been a method of choice for expressing privacy guarantees of a variety of differentially private algorithms, especially those that rely on the Gaussian additive noise mechanism or whose analysis follows from composition theorems. The additive $\delta$ parameter allows suppressing the long tails of the mechanism's distribution where pure \eps-differential privacy guarantees may not hold.

Compared to the standard definition, \epsdelta-differential privacy offers asymptotically smaller cumulative loss under composition and allows greater flexibility in the selection of privacy-preserving mechanisms.

Despite its notable advantages and numerous applications, the definition of \epsdelta-differential privacy is an imperfect fit for its two most common use cases: the Gaussian mechanism and a composition rule. We briefly sketch them here and elaborate on these points in the next section.

The first application of \epsdelta-differential privacy was the analysis of the Gaussian noise mechanism~\cite{ODO}. In contrast with the Laplace mechanism, whose privacy guarantee is characterized tightly and accurately by \eps-differential privacy, a single Gaussian mechanism satisfies a \emph{curve} of  $(\eps(\delta), \delta)$-differential privacy definitions. Picking any one point on this curve leaves out important information about the mechanism's actual behavior.

The second common use of \epsdelta-differential privacy is due to applications of advanced composition theorems. The central feature of 
 \eps-differential privacy is that it is closed under composition; moreover, the \eps\ parameters of composed mechanisms simply add up, which motivates the concept of a \emph{privacy budget}.  By relaxing the guarantee to \epsdelta-differential privacy, advanced composition allows tighter analyses for compositions of (pure) differentially private mechanisms. Iterating this process, however, quickly leads to a combinatorial explosion of parameters, as each application of an advanced composition theorem leads to a wide selection of possibilities for $(\eps(\delta), \delta)$-differentially private guarantees.
 
In part to address the shortcomings of \epsdelta-differential privacy,  several recent works, surveyed in the next section, explored the use of  higher-order \emph{moments} as a way of bounding the tails of the privacy loss variable.

Inspired by these theoretical results and their applications, we propose \emph{\renyi differential privacy} as a natural relaxation of differential privacy that is well-suited for expressing guarantees of privacy-preserving algorithms and for composition of heterogeneous mechanisms. Compared to \epsdelta-differential privacy, \renyi differential privacy is a strictly stronger privacy definition. It offers an operationally convenient and quantitatively accurate way of tracking cumulative privacy loss throughout execution of a standalone differentially private mechanism and across many such mechanisms. Most significantly, \renyi differential privacy allows combining the intuitive and appealing concept of a privacy budget with application of advanced composition theorems.

The paper presents a self-contained exposition of the new definition, unifying current literature and demonstrating its applications. The organization of the paper is as follows. Section~\ref{s:dp} reviews the standard definition of differential privacy, its \epsdelta relaxation and its most common uses. Section~\ref{s:renyi} introduces the definition of \renyi differential privacy and proves its basic properties that parallel those of \eps-differential privacy, summarizing the results in Table~\ref{tab:properties}. Section~\ref{s:rdp_epsdelta} demonstrates a reduction from \renyi differential privacy to \epsdelta-differential privacy, followed by a proof of an advanced composition theorem in Section~\ref{s:advanced}. Section~\ref{s:basic} applies \renyi differential privacy to analysis of several basic mechanisms: randomized response for predicates, Laplace and Gaussian (see Table~\ref{tab:mechanisms} for a brief summary). Section~\ref{s:discussion} discusses assessment of risk due to application of a \renyi differentially private mechanism and use of \renyi differential privacy as a privacy loss tracking tool. Section~\ref{s:conclusions} concludes with open questions.

\section{Differential Privacy and Its Flavors}\label{s:dp}

\leader{\eps-Differential privacy~\cite{DMNS06}.} We first recall the standard definition of \eps-differential privacy.
\begin{definition}[\eps-DP]
	A randomized mechanism $f\colon \D\mapsto\R$ satisfies \eps-differential privacy (\eps-DP) if for any adjacent $D,D'\in\D$ and $S\subset \R$
	\[\Pr[f(D)\in S]\leq e^\eps\Pr[f(D')\in S].\]
\end{definition}
The above definition is contingent on the notion of \emph{adjacent} inputs $D$ and $D'$, which is domain-specific and is typically chosen to capture the contribution to the mechanism's input by a single individual.

The \emph{Laplace} mechanism is a prototypical \eps-differentially private algorithm, allowing release of an approximate (noisy) answer to an arbitrary query with values in~$\mathbb{R}^n$. The mechanism is defined as
\[
\LaplaceM_\eps f(x) \triangleq f(x) + \Lambda(0,\Delta_1 f/\eps),
\]
where $\Lambda$ is the Laplace distribution and $\ell_1$-sensitivity of the query $f$ is
\[
\Delta_1 f \triangleq \max_{D,D'} \|f(D)-f(D')\|_1
\]
taken over all adjacent inputs $D$ and $D'$.

The basic composition theorem states that if $f$ and $g$ are, respectively, $\eps_1$- and $\eps_2$-DP, then the simultaneous release of $f(D)$ and $g(D)$ satisfies $(\eps_1+\eps_2)$-DP. Moreover, the mechanism $g$ may be selected adaptively, after seeing the output of $f(D)$.

\leader{\epsdelta-Differential privacy~\cite{ODO}.} A relaxation of \eps-differential privacy allows a $\delta$ additive term in its defining inequality:
\begin{definition}[\epsdelta-DP]
	A randomized mechanism $f\colon \D\mapsto\R$ offers \epsdelta-differential privacy if for any adjacent $D,D'\in\D$ and $S\subset \R$
	\[\Pr[f(D)\in S]\leq e^\eps\Pr[f(D')\in S]+\delta.\]
\end{definition}

The common interpretation of \epsdelta-DP is that it is \eps-DP ``except with probability $\delta$''. Formalizing this statement runs into difficulties similar to the ones addressed by Mironov et al.~\cite{MPRV09-CDP} for a different (computational) relaxation. For any two adjacent inputs, $D_1$ and $D_2$, it is indeed possible to define an $\eps$-DP mechanism that agrees with $f$ with all but $\delta$ probability. Extending this argument to domains of exponential sizes (for instance, to a boolean hypercube) cannot be done without diluting the guarantee exponentially~\cite{De12}. We conclude that \epsdelta-differential privacy is a \emph{qualitatively} different definition than pure \eps-DP (unless, of course, $\delta=0$, which we assume not to be the case through the rest of this section).

Even for the simple case of exactly two input databases (such as when the adversary knows the entire dataset except whether it contains a particular record), the $\delta$ additive term encompasses two very different modes in which privacy may fail. In both scenarios \eps-DP holds with probability $1-\delta$, they differ in what happens with the remaining probability~$\delta$. In the first scenario privacy degrades gracefully, such as to $\eps_1$-DP with probability $\delta/2$, to $\eps_2$-DP with probability $\delta/4$, etc. In the other scenario, with probability $\delta$ the secret---whether the record is part of the database or not---becomes completely exposed. The difference between the two failure modes can be quite stark. In the former, there is always some residual deniability; in the latter, the adversary occasionally learns the secret with certainty. Depending on the adversary's tolerance to false positives, plausible deniability may offer adequate protection, but a single \epsdelta-DP privacy statement cannot differentiate between the two alternatives. For a lively parable  of the different guarantees offered by the \eps-DP and \epsdelta-DP definitions see McSherry~\cite{McSherry17-lottery}.

To avoid the worst-case scenario of always violating privacy of a $\delta$ fraction of the dataset, the standard recommendation is to choose $\delta\ll 1/N$ or even $\delta=\mathrm{negl}(1/N)$, where $N$ is the number of contributors. This strategy forecloses possibility of one particularly devastating outcome, but other forms of information leakage remain.

The definition of \epsdelta-differential privacy was initially proposed to capture privacy guarantees of the Gaussian mechanism, defined as follows:
\[
\GaussianM_\sigma f(x) \triangleq f(x) + N(0, \sigma^2).
\]
Elementary analysis shows that the Gaussian mechanism cannot meet \eps-DP for any $\eps$. Instead, it satisfies a continuum of incomparable \epsdelta-DP guarantees, for all combinations of $\eps<1$ and $\sigma>\sqrt{2 \ln{1.25/\delta}}\Delta_2 f/\eps$, where $f$'s $\ell_2$-sensitivity is defined as
\[
\Delta_2 f \triangleq \max_{D,D'} \|f(D)-f(D')\|_2
\]
taken over all adjacent inputs $D$ and $D'$.

There exist valid reasons for preferring the Gaussian mechanism over Laplace: the noise comes from the same Gaussian distribution (closed under addition) as the error that may already be present in the dataset; the standard deviation of the noise is proportional to the  query's $\ell_2$ sensitivity, which is no larger and often much smaller than~$\ell_1$; for the same standard deviation, the tails of the Gaussian (normal) distribution decay much faster than those of the Laplace (exponential) distribution. Unfortunately, distilling the guarantees of the Gaussian mechanism down to a single number or a small set of numbers using the language of \epsdelta-DP always leaves a possibility of a complete privacy compromise that the mechanism itself may not allow.

Another common reason for bringing in \epsdelta-differential privacy is application of advanced composition theorems. Consider the case of $k$-fold adaptive composition of an \epsdelta-DP mechanism. For any $\delta'>0$ it holds that the composite mechanism is $(\eps',k\delta+\delta')$-DP, where $\eps'\triangleq\sqrt{2k\ln(1/\delta')}\eps+k\eps(e^\eps-1)$~\cite{DRV10-boosting}. Note that, similarly to our discussion of the Gaussian mechanism, a single mechanism satisfies a continuum of incomparable  \epsdelta-DP guarantees.

Kairouz et al. give a procedure for computing an \emph{optimal} $k$-fold composition of an \epsdelta-DP mechanism~\cite{KOV15-composition}. Murtagh and Vadhan~\cite{MV16-optimal} demonstrate that generalizing this result to composition of heterogeneous mechanisms (i.e., satisfying $(\eps_i,\delta_i)$-DP for different $\eps_i$'s) is \#P-hard; they describe a PTAS for an approximate solution. None of these works tackles the problem of composing mechanisms that satisfy several \epsdelta-DP guarantees simultaneously.

\leader{(zero)-Concentrated Differential Privacy and the moments accountant.} The closely related work by Dwork and Rothblum~\cite{DR16-CDP}, followed by Bun and Steinke~\cite{BS16-zCDP}, explore privacy definitions---Concentrated Differential Privacy and zero-Concentrated Differential Privacy---that are framed using the language of, respectively, subgaussian tails and the \renyi divergence. The main difference between our approaches is that both Concentrated and zero-Concentrated DP require a linear bound on \emph{all} positive moments of a privacy loss variable. In contrast, our definition applies to one moment at a time. Although less restrictive, it allows for more accurate numerical analyses.

The work by Abadi et al.~\cite{Abadi16-DP-DL} on differentially private stochastic gradient descent introduced the \emph{moments accountant} as an internal tool for tracking privacy loss across multiple invocations of the Gaussian mechanism applied to random subsets of the input dataset.  The paper's results are expressed via a necessarily lossy translation of the accountant's output (bounds on select moments of the privacy loss variable) to the language of \epsdelta-differential privacy. 

Taken together, the works on Concentrated DP, zero-Concentrated DP, and the moments accountant point towards adopting \renyi differential privacy as an effective and flexible mechanism for capturing privacy guarantees of a wide variety of algorithms and their combinations.

\leader{Other relaxations.} We briefly mention other relaxations and generalizations of differential privacy.

Under the indistinguishability-based Computational Differential Privacy (IND-CDP) definition~\cite{MPRV09-CDP}, the test of closeness between distributions on adjacent inputs is computationally bounded (all other definitions considered in this paper hold against an unbounded, information-theoretic adversary). The IND-CDP notion allows much more accurate functionalities in the two-party setting~\cite{MMPRTV10}; in the traditional client-server setup there is a natural class of functionalities where the gap between IND-CDP and \epsdelta-DP is minimal~\cite{GKY11-limits}, and there are (contrived) examples where the computational relaxation permits tasks that are infeasible under information-theoretic definitions~\cite{BCV16-separating}.

Several other works, most notably the Pufferfish and the coupled-worlds frameworks~\cite{KM14-pufferfish,BGKS13-coupled}, propose different stability constraints on the output distribution of privacy-preserving mechanisms. Although they differ in \emph{what} distributions are compared, their notion of closeness is the same as in~\epsdelta-DP.

\section{\renyi differential privacy}\label{s:renyi}

We describe a generalization of the notion of differential privacy based on the concept of the \renyi\ divergence. Connection between the two notions has been pointed out before (mostly for one extreme order, known as the Kullback-Leibler divergence~\cite{DRV10-boosting,DJW13-local}); our contribution is in systematically exploring the relationship and its practical implications. 

The (parameterized) \renyi\ divergence is classically defined as follows~\cite{Renyi61}:
\begin{definition}[\renyi divergence]
For two probability distributions $P$ and $Q$ defined over $\R$, the \renyi\ divergence of order $\alpha>1$ is
\begin{align*}
D_\alpha(P\|Q)\triangleq\frac{1}{\alpha-1}\log \E_{x\sim Q}\left( \frac{P(x)}{Q(x)}\right)^\alpha.
\end{align*}
\end{definition}
(All logarithms are natural; $P(x)$ is the density of $P$ at $x$.)

For the endpoints of the interval $(1,\infty)$ the \renyi\ divergence is defined by continuity. Concretely, $D_1(P\|Q)$ is set to be $\lim_{\alpha\to1}D_\alpha(P\|Q)$ and can be verified to be equal to the Kullback-Leibler divergence (also known as relative entropy):
\[
D_1(P\|Q)=\E_{x\sim P}\log\frac{P(x)}{Q(x)}.
\]
Note that the expectation is taken over $P$, rather than over $Q$ as in the previous definition. It is possible, though, that $D_1(P\|Q)$ thus defined is finite whereas $D_\alpha(P\|Q)=+\infty$ for all $\alpha>1$.

Likewise,
\[D_\infty(P\|Q)=\sup_{x\in \supp{Q}}\log \frac{P(x)}{Q(x)}.\]

For completeness, we reproduce in the Appendix properties of the \renyi divergence important to the sequel: non-negativity, monotonicity, probability preservation, and a weak triangle inequality (Propositions~\ref{prop:non-negativity}--\ref{prop:triangle}).

The relationship between the \renyi\ divergence with $\alpha=\infty$ and differential privacy is immediate. A randomized mechanism $f$ is \eps-differentially private if and only if its distribution over any two adjacent inputs $D$ and $D'$ satisfies 
\[D_\infty\left(f(D)\|f(D')\right)\leq \eps.\]

It motivates exploring a relaxation of differential privacy based on the \renyi\ divergence. 
\begin{definition}[$(\alpha, \eps)$-RDP]
A randomized mechanism $f\colon \D\mapsto\R$ is said to have \eps-\renyi\ differential privacy of order $\alpha$, or $(\alpha, \eps)$-RDP for short, if for any adjacent $D,D'\in\D$ it holds that 
\[D_\alpha\left(f(D)\|f(D')\right)\leq \eps.\]
\end{definition}

\begin{remark}
Similarly to the definition of differential privacy, a finite value for \eps-RDP implies that feasible outcomes of $f(D)$ for \emph{some} $D\in\D$ are feasible, i.e., have a non-zero density, for all inputs from $\D$ except for a set of measure 0. Assuming that this is the case, we let the event space be the support of the distribution.
\end{remark}


\begin{remark}
The \renyi\ divergence can be defined for $\alpha$ smaller than 1, including negative orders. We are not using these orders in our definition of \renyi\ differential privacy.
\end{remark}

The standard definition of differential privacy has been successful as a privacy measure because it simultaneously meets several important criteria. We verify that the relaxed definition inherits many of the same properties. The results of this section are summarized in Table~\ref{tab:properties}.

\leader{``Bad outcomes" guarantee.} A privacy definition is only as useful as its guarantee for data contributors. The simplest such assurance is the ``bad outcomes'' interpretation. Consider a person, concerned about some adverse consequences, deliberating whether to withhold her record from the database. Let us say that some outputs of the mechanism are labeled as ``bad.'' The differential privacy guarantee asserts that the probability of observing a bad outcome will not change (either way) by more than a factor of $e^\eps$ whether anyone's record is part of the input or not (for appropriately defined ``adjacent'' inputs). This is an immediate consequence of the definition of differential privacy, where the subset $S$ is the union of bad outcomes.

This guarantee is relaxed for \renyi\ differential privacy. Concretely, if $f$ is $(\alpha,\eps)$-RDP, then by Proposition~\ref{prop:preservation}:
\begin{align*}
e^{-\eps}\Pr[f(D')\in S]^{\alpha/(\alpha-1)}&\leq\Pr[f(D)\in S]\\
&\leq\left(e^{\eps}\Pr[f(D')\in S]\right)^{(\alpha-1)/\alpha}.
\end{align*}

We discuss consequences of this relaxation in Section~\ref{s:discussion}. 

\leader{Robustness to auxiliary information.} Critical to the adoption of differential privacy as an operationally useful definition is its lack of assumptions on the adversary's knowledge. More formally, the property is captured by the Bayesian interpretation of privacy guarantees, which compares the adversary's prior with the posterior.

Assume that the adversary has a prior $p(D)$ over the set of possible inputs $D\in\D$, and observes an output $X$ of an \eps-differentially private mechanism $f$. Its posterior satisfies the following guarantee for all pairs of adjacent inputs $D,D'\in\D$ and all $X\in \R$:
\[
\frac{p(D\hphantom{'}|X)}{p(D'|X)}\leq e^{\eps}\frac{p(D)}{p(D')}.
\]
In other words, evidence obtained from an \eps-differentially private mechanism does not move the relative probabilities assigned to adjacent inputs (the odds ratio) by more than~$e^{\eps}$.

The guarantee implied by RDP is a probabilistic statement about the change in the Bayes factor. Let the random variable $R(D,D')$ be defined as follows:
\begin{multline*}
R(D,D')\sim \frac{p(D'|X)}{p(D\hphantom{'}|X)}= \frac{p(X|D')\cdot p(D')}{p(X|D\hphantom{'})\cdot p(D\hphantom{'})},\\
\;\textrm{where }X\sim f(D).
\end{multline*}

It follows immediately from definition that the \renyi\ divergence of order $\alpha$ between $P=f(D')$ and $Q=f(D)$ bounds the $\alpha$-th moment of the \emph{change} in $R$:
\begin{multline*}
\E_Q\left[\left\{\frac{R_\textrm{post}(D,D')}{R_\textrm{prior}(D,D')}\right\}^\alpha\right]=\E_Q\left[P(x)^\alpha Q(x)^{-\alpha}\right]=\\
\exp[(\alpha-1)D_\alpha(f(D')\|f(D))].
\end{multline*} 

By taking the logarithm of both sides and applying Jensen's inequality we obtain that
\begin{multline}\label{eq:posterior-bound}
\E_{f(D)}\left[\log R_\textrm{post}(D,D')-\log R_\textrm{prior}(D,D')\right]\leq\\
 D_\alpha(f(D)\|f(D')).
\end{multline}
(This can also be derived by observing that
\begin{multline*}
\E_{f(D)}\left[\log R_\textrm{post}(D,D')-\log R_\textrm{prior}(D,D')\right]=\\
D_1(f(D)\|f(D'))
\end{multline*}
 and by monotonicity of the \renyi\ divergence.)

Compare (\ref{eq:posterior-bound}) with the guarantee of pure differential privacy, which states that $\log R_\textrm{post}(D,D')- \log R_\textrm{prior}(D,D')\leq \eps$ everywhere, not just in expectation.

\leader{Post-processing.} A privacy guarantee that can be diminished by manipulating output is unlikely to be useful. Consider a randomized mapping $g\colon \R\mapsto\R'$. We observe that $D_\alpha(P\|Q)\geq D_\alpha(g(P)\|g(Q))$ by the analogue of the data processing inequality~\cite[Theorem 9]{EH07-Renyi}. It means that if $f(\cdot)$ is \alphaeps-RDP, so is $g(f(\cdot))$. In other words, \renyi\ differential privacy is preserved by  post-processing.

\leader{Preservation under adaptive sequential composition.} The property that makes possible modular construction of differentially private algorithms is self-composition: if $f(\cdot)$ is $\eps_1$-differentially private and $g(\cdot)$ is $\eps_2$-differentially private, then simultaneous release of $f(D)$ and $g(D)$ is $\eps_1+\eps_2$-differentially private. The guarantee even extends to when $g$ is chosen \emph{adaptively} based on $f$'s output: if $g$ is indexed by elements of $\R$ and $g_X(\cdot)$ is $\eps_2$-differentially private for any $X\in \R$, then publishing $(X, Y)$, where $X\sim f(D)$ and $Y\sim g_X(D)$, is $\eps_1+\eps_2$-differentially private.

We prove a similar statement for the composition of two RDP mechanisms.

\begin{proposition}\label{prop:self-composition}Let $f\colon \D\mapsto\R_1$ be $(\alpha,\eps_1)$-RDP and $g\colon \R_1\times\D\mapsto\R_2$ be $(\alpha,\eps_2)$-RDP, then the mechanism defined as $(X,Y)$, where $X\sim f(D)$ and $Y\sim g(X,D)$, satisfies $(\alpha,\eps_1+\eps_2)$-RDP.
\end{proposition}

\begin{proof}
Let $h\colon \D\mapsto\R_1\times\R_2$ be the result of running $f$ and $g$ sequentially. We write $X$, $Y$, and $Z$ for the distributions $f(D)$, $g(X, D)$, and the joint distribution $(X,Y)=h(D)$. $X'$, $Y'$, and $Z'$ are similarly defined if the input is $D'$.  Then
\begin{align*}
\exp&\left[(\alpha-1)D_\alpha(h(D)\|h(D'))\right]\\
&=\int_{\R_1\times\R_2}Z(x,y)^\alpha Z'(x,y)^{1-\alpha}\dx\dy\\
&=\int_{\R_1}\int_{\R_2}(X(x)Y(x,y))^\alpha(X'(x)Y'(x,y))^{1-\alpha}\dy\dx\\
&=\int_{\R_1}\!\!\!\!X(x)^\alpha X'(x)^{1-\alpha}\left\{\int_{\R_2}\!\!\!\!Y(x,y)^\alpha Y'(x,y)^{1-\alpha}\dy\right\}\!\dx\\
&\leq \int_{\R_1}X(x)^\alpha X'(x)^{1-\alpha}\dx\cdot \exp((\alpha-1)\eps_2)\\
&\leq \exp((\alpha-1)\eps_1)\exp((\alpha-1)\eps_2)\\
&=\exp((\alpha-1)(\eps_1+\eps_2)),
\end{align*}
from which the claim follows.
\end{proof}

Significantly, the guarantee holds whether the releases of $f$ and $g$ are coordinated or not, or computed over the same or different versions of the input dataset. It allows us to operate with a well-defined notion of a \emph{privacy budget} associated with an individual, which is a finite resource consumed with each differentially private data release.

Extending the concept of the privacy budget, we say that the \renyi\ differential privacy has a \emph{budget curve} parameterized by the order $\alpha$. We present examples illustrating this viewpoint in Section~\ref{s:basic}.

\leader{Group privacy.} Although the definition of differential privacy constrains a mechanism's outputs on pairs of \emph{adjacent} inputs, its guarantee extends, in a progressively weaker form, to inputs that are farther apart. This property has two important consequences. First, the differential privacy guarantee degrades gracefully if our assumptions about one person's influence on the input are (somewhat) wrong. For example, a single family contributing to a survey will likely share many socio-economic, demographic, and health characteristics. Rather than collapsing, the differential privacy guarantee will scale down linearly with the number of family members. Second, the group privacy property allows preprocessing input into a differentially private mechanism, possibly amplifying (in a controlled fashion) one record's impact on the output of the computation.

We define group privacy using a notion of $c$-stable transformation~\cite{McSherry09-PINQ}. We say that $g\colon\D\mapsto\D'$ is $c$-stable if $g(A)$ and $g(B)$ are adjacent in $D'$ implies that there exists a sequence of length $c+1$ so that $D_0=A,\dots,D_c=B$ and all $(D_i,D_{i+1})$ are adjacent in $\D$.

The standard notion of differential privacy satisfies the following. If $f$ is \eps-differentially private and $g\colon \D'\mapsto\D$ is $c$-stable, then $f\circ g$ is $c\eps$-differentially private. A similar statement holds for \renyi\ differential privacy.

\begin{proposition}\label{prop:group}If $f\colon \D\mapsto\R$ is \alphaeps-RDP, $g\colon \D'\mapsto\D$ is $2^c$-stable and $\alpha\geq2^{c+1}$, then $f\circ g$ is $(\alpha/2^c,3^c\eps)$-RDP.
\end{proposition}

\begin{proof}We prove the statement for $c=1$, the rest follows by induction. 

Define $h=f\circ g$. Since $g$ is 2-stable, it means that for any adjacent $D,D'\in \D'$ there exist $A\in D$, so that $g(D)$ and $A$, $A$ and $g(D')$ are adjacent in $D$.

By Corollary~\ref{col:triangle} and monotonicity of the \renyi\ divergence, we have that $h=f\circ g$ satisfies
\begin{multline*}
D_{\alpha/2}(h(D)\|h(D'))\leq\frac{\alpha-1}{\alpha-2}D_{\alpha}(h(D)\|h(A))+\\
D_{\alpha-1}(h(A)\|h(D'))\leq 3\eps.
\end{multline*}
\end{proof}

\begin{table*}
\begin{center}
\begin{tabular}{p{.28\textwidth}p{.31\textwidth}p{.33\textwidth}}
\toprule
\textbf{Property} & \textbf{Differential Privacy} & \textbf{\renyi\ Differential Privacy}\\
\midrule
\multirow{2}{.26\textwidth}{Change in probability of outcome $S$} & $\Pr[f(D)\in S]\leq\ e^{\eps}\Pr[f(D')\in S]$ &$\Pr[f(D)\in S]\leq\left(e^{\eps}\Pr[f(D')\in S]\right)^{(\alpha-1)/\alpha}$\\
& $\Pr[f(D)\in S]\geq e^{-\eps}\Pr[f(D')\in S]$ &$\Pr[f(D)\in S]\geq e^{-\eps}\Pr[f(D')\in S]^{\alpha/(\alpha-1)}$\\
\midrule
Change in the Bayes factor& $\dfrac{R_\textrm{post}(D,D')}{R_\textrm{prior}(D,D')}\leq e^\eps$ always
& $\E\left[\left\{\dfrac{R_\textrm{post}(D,D')}{R_\textrm{prior}(D,D')}\right\}^\alpha\right]\leq \exp[(\alpha-1)\eps]$ \\
\midrule
Change in log of the Bayes factor & $|\Delta \log R(D,D')|\leq \eps$ always
& $\E[\Delta \log R(D,D')]\leq \eps$
\\
\midrule
Post-processing & \multicolumn{2}{c}{$f$ is \eps-DP (or $(\alpha, \eps)$-RDP) $\Rightarrow$ $g\circ f$ is \eps-DP (or $(\alpha, \eps)$-RDP, resp.)}\\ 
\midrule
Adaptive sequential composition (basic)& \multicolumn{2}{c}{$f,g$ are \eps-DP (or $(\alpha, \eps)$-RDP) $\Rightarrow$ $(f,g)$ is $2\eps$-DP (resp., $(\alpha, 2\eps)$-RDP)}\\
\midrule
Group privacy, pre-processing & \multicolumn{2}{c}{$f$ is \eps-DP (or $(\alpha, \eps)$-RDP), $g$ is $2^c$-stable $\Rightarrow$ $f\circ g$ is $2^c\eps$-DP (resp., $(\alpha/2^c,3^c\eps)$-RDP)}\\
\bottomrule
\end{tabular}
\end{center}
\caption{Summary of properties shared by differential privacy and RDP.}\label{tab:properties}
\end{table*}

\section{RDP and \epsdelta-DP}\label{s:rdp_epsdelta}

As we observed earlier, the definition of \eps-differential privacy coincides with $(\infty,\eps)$-RDP. By monotonicity of the \renyi\ divergence, $(\infty,\eps)$-RDP implies $(\alpha,\eps)$-RDP for all finite~$\alpha$.

In turn, an \alphaeps-RDP implies $(\eps_\delta,\delta)$-differential privacy for any given probability $\delta >0$.

\begin{proposition}[From RDP to \epsdelta-DP]\label{prop:rdp-to-epsdelta}
If $f$ is an \alphaeps-RDP mechanism, it also satisfies $(\eps+\frac{\log 1/\delta}{\alpha-1},\delta)$-differential privacy for any $0<\delta<1$.
\end{proposition}
\begin{proof}
	
	
	Take any two adjacent inputs $D$ and $D'$, and a subset of $f$'s range $S$. To show that $f$ is $(\eps',\delta)$-differentially private, where $\eps'=\eps+\frac1{\alpha-1}\log 1/\delta$, we need to demonstrate that $\Pr[f(D)\in S]\leq e^{\eps'}\Pr[f(D')\in S]+\delta$. In fact, we prove a stronger statement that $\Pr[f(D)\in S]\leq \max(e^{\eps'}\Pr[f(D')\in S], \delta)$.
	
	Recall that by Proposition~\ref{prop:preservation}
	\[\Pr[f(D)\in S] \leq \{e^\eps \Pr[f(D')\in S\}^{1-1/\alpha}.
	\]
	Denote $\Pr[f(D')\in S]$ by $Q$ and consider two cases.
	
	Case I. $e^\eps Q> \delta^{\alpha/(\alpha-1)}$. Continuing the above,
	\begin{align*}
	\Pr[f(D)\in S]&\leq \{e^\eps Q\}^{1-1/\alpha}=e^\eps Q\cdot \{e^\eps Q\}^{-1/\alpha}\\
	&\leq e^\eps Q\cdot \delta^{-1/(\alpha-1)}\\
	&=\exp\left(\eps+\frac{\log 1/\delta}{\alpha-1}\right)\cdot Q.
	\end{align*}
	
	Case II. $e^\eps Q\leq \delta^{\alpha/(\alpha-1)}$.
	This case is immediate since 
	\[
	\Pr[f(D)\in S]\leq \{e^\eps Q\}^{1-1/\alpha}\leq \delta,
	\]
which completes the proof.	
\end{proof}

A more detailed comparison between the notions of RDP and \epsdelta-differential privacy that goes beyond these reductions is deferred to Section~\ref{s:discussion}.

\section{Advanced Composition Theorem}\label{s:advanced}

The main thesis of this section is that the \renyi\ differential privacy curve of a composite mechanism is sufficient to draw non-trivial conclusions about its privacy guarantees, similar to the ones given by other advanced composition theorems, such as Dwork et al.~\cite{DRV10-boosting} or Kairouz et al.~\cite{KOV15-composition}. Although our proof is structured similarly to~Dwork et al.\ (for instance, Lemma~\ref{lemma:D_alpha_bound} is a direct generalization of~\cite[Lemma III.2]{DRV10-boosting}), it is phrased entirely in the language of \renyi\ differential privacy without making any (explicit) use of probability arguments.

\begin{lemma}\label{lemma:D_alpha_bound}
If $P$ and $Q$ are such that $D_\infty(P\|Q)\leq \eps$ and $D_\infty(Q\|P)\leq \eps$, then for $\alpha\geq 1$
\[
D_\alpha(P\|Q) \leq 2\alpha \eps^2.
\]
\end{lemma}
\begin{proof}
If $\alpha\geq 1+1/\eps$, then 
\[
D_\alpha(P\|Q) \leq D_\infty(P\|Q)=\eps\leq (\alpha-1)\eps^2.
\]

Consider the case when $\alpha<1+1/\eps$. We first observe that for any $x>y>0$, $\lambda=\log (x/y)$, and $0\leq \beta\leq 1/\lambda$ the following inequality holds:
\begin{multline}\label{eq:beta-lambda-bound}
x^{\beta+1} y^{-\beta}+x^{-\beta} y^{\beta+1}=x\cdot e^{\beta\lambda}+y\cdot e^{-\beta\lambda}\\\leq x(1+\beta\lambda+(\beta\lambda)^2)+y(1-\beta\lambda+(\beta\lambda)^2)\\
=(1+(\beta\lambda)^2)(x+y)+\beta\lambda(x-y).
\end{multline}
Since all terms of the right hand side of (\ref{eq:beta-lambda-bound}) are positive, the inequality applies if $\lambda$ is an upper bound on $\log x/y$, which we use in the argument below.

\begin{align*}
\exp&[(\alpha-1)D_\alpha(P\|Q)]\\
&=\int_\R P(x)^\alpha Q(x)^{1-\alpha}\dx\\
&\leq \int_\R \left\{P(x)^\alpha Q(x)^{1-\alpha}+Q(x)^\alpha P(x)^{1-\alpha}\right\}\dx-1\tag{by nonnegativity of $D_\alpha(Q\|P)$}\\
&\leq \int_\R\big\{(1+(\alpha-1)^2\eps^2)(P(x)+Q(x))+\\
&\quad\quad\quad(\alpha-1)\eps|P(x)-Q(x)|\big\}\dx-1\tag{by (\ref{eq:beta-lambda-bound}) for $\beta=\alpha-1\leq 1/\eps$}\\
&= 1+2(\alpha-1)^2\eps^2+(\alpha-1)\eps\|P-Q\|_1.\\
\end{align*}
Taking the logarithm of both sides and using that $\log(1+x)<x$ for positive $x$ we find that
\begin{equation}\label{eq:D_TV}
D_\alpha(P\|Q)\leq 2(\alpha-1)\eps^2+\eps\|P-Q\|_1.
\end{equation}

Observe that 
\begin{multline*}
\|P-Q\|_1=\int\left|P(x)-Q(x)\right|\dx\\
=\int_\R\min(P(x),Q(x))\left|\frac{\max(P(x),Q(x))}{\min(P(x),Q(x))}-1\right|\dx\\
\leq \min(2,e^\eps-1)\leq 2\eps^2.
\end{multline*}
Plugging the bound on $\|P-Q\|_1$ into~(\ref{eq:D_TV}) completes the proof.

The claim for $\alpha=1$ follows by continuity.
\end{proof}

The constant in Lemma~\ref{lemma:D_alpha_bound} can be improved to .5 via a substantially more involved analysis~\cite[Proposition 3.3]{BS16-zCDP} (see also )

\begin{proposition}\label{prop:advanced_comp}
Let $f\colon\D\mapsto\R$ be an adaptive composition of $n$ mechanisms all satisfying \eps-differential privacy. Let $D$ and $D'$ be two adjacent inputs. Then for any $S\subset\R$:
\begin{multline*}
\Pr[f(D)\in S] \leq
 \exp\left(2\eps\sqrt{n\log1/\Pr[f(D')\in S]}\right)\\
 \cdot \Pr[f(D')\in S].
\end{multline*}
\end{proposition}
\begin{proof} 

By applying Lemma~\ref{lemma:D_alpha_bound} to the \renyi\ differential privacy curve of the underlying mechanisms and Proposition~\ref{prop:self-composition} to their composition, we find that for all $\alpha\geq 1$
\[
D_\alpha(f(D)\|f(D'))\leq 2\alpha n \eps^2.
\]

Denote $\Pr[f(D')\in S]$ by $Q$ and consider two cases.

Case I: $\log 1/Q\geq \eps^2n$. Choosing with some foresight $\alpha=\sqrt{\log 1/Q}/(\eps\sqrt{n})\geq 1$, we have by Proposition~\ref{prop:preservation} (probability preservation):
\begin{align*}
\Pr[f(D)\in S]&\leq \left\{\exp[D_\alpha(f(D)\|f(D')]\cdot Q\right\}^{1-1/\alpha}\\
&\leq \exp(2(\alpha-1) n \eps^2)\cdot Q^{1-1/\alpha}\\
&< \exp\left(\eps\sqrt{n\log1/Q} - (\log Q)/\alpha\right)\cdot Q\\
&= \exp\left(2\eps\sqrt{n\log1/Q}\right)\cdot Q.
\end{align*}

Case II: $\log 1/Q< \eps^2n$. This case follows trivially, since the right hand side of the claim is larger than 1:
\[
\exp\left(2\eps\sqrt{n\log1/Q}\right)\cdot Q\geq \exp\left(2\log 1/Q\right)\cdot Q=1/Q>1.
\]
\end{proof}

The notable feature of Proposition~\ref{prop:advanced_comp} is that its privacy guarantee---bounded probability gain---comes in the form that depends on the event's probability. We discuss this type of guarantee in Section~\ref{s:discussion}.

The following corollary gives a more conventional \epsdelta\ variant of advanced composition.

\begin{corollary}\label{cor:advanced_comp_epsdelta}
Let $f$ be the composition of the $n$ $\eps$-differentially private mechanisms. Let $0<\delta<1$ be such that $\log(1 / \delta) \geq \eps^2 n$. Then $f$ satisfies $(\eps',\delta)$-differential privacy where
\[
\eps'\triangleq 4\eps\sqrt{2 n \log(1 / \delta)}.
\]
\end{corollary}
\begin{proof}
Let $D$ and $D'$ be two adjacent inputs, and $S$ be some subset of the range of $f$. To argue $(\eps',\delta)$-differential privacy of $f$, we need to verify that 
\[
\Pr[f(D)\in S]\leq e^{\eps'}\Pr[f(D')\in S]+\delta.
\]
In fact, we prove a somewhat stronger statement, namely that $\Pr[f(D)\in S]\leq \max(e^{\eps'}\Pr[f(D')\in S], \delta).$

By Proposition~\ref{prop:advanced_comp}
\begin{multline*}
\Pr[f(D)\in S] \leq
 \exp\left(2\eps\sqrt{n\log1/\Pr[f(D')\in S]}\right)\\
 \cdot \Pr[f(D')\in S].
\end{multline*}

Denote $\Pr[f(D')\in S]$ by $Q$ and consider two cases.

Case I: $8\log 1/\delta > \log 1/Q$.  Then
\begin{align*}
\Pr[f(D)\in S] &\leq \exp\left(2\eps\sqrt{n \log 1/Q}\right)\cdot Q\\
&< \exp\left(2\eps\sqrt{8n\log 1/\delta}\right)\cdot Q\tag{by $8\log 1/\delta > \log 1/Q$}\\
&= \exp\left(\eps'\right)\cdot Q.
\end{align*}

Case II: $8\log 1/\delta \leq \log 1/Q$.  Then
\begin{align*}
\Pr[f(D)\in S] &\leq \exp\left(2\eps\sqrt{n \log 1/Q}\right)\cdot Q\\
&\leq \exp\left(2\sqrt{\log 1/\delta\cdot \log 1/Q}\right)\cdot Q\tag{since $\log(1 / \delta) \geq \eps^2 n$}\\
&\leq \exp\left(\sqrt{1/2}\log 1/Q\right)\cdot Q\tag{since $8\log 1/\delta \leq \log 1/Q$}\\
&=Q^{1-1/\sqrt{2}}\leq Q^{1/8}\\
&\leq \delta\tag{ditto}.
\end{align*}
\end{proof}

\begin{remark}
The condition $\log(1 / \delta) \geq \eps^2 n$ corresponds to the so-called ``high privacy'' regime of the advanced composition theorem~\cite{KOV15-composition}, where $\eps'<(1+\sqrt{2})\log(1/\delta)$. Since $\delta$ is typically chosen to be small, say, less than 1\%, it covers the case of $\eps'<11$. In other words, if $\log(1 / \delta) > \eps^2 n$, this and other composition theorems are unlikely to yield strong bounds.
\end{remark}

\section{Basic Mechanisms}\label{s:basic}

In this section we analyze \renyi\ differential privacy of three basic mechanisms and their self-composition: randomized response, Laplace and Gaussian noise addition. The results are summarized in Table~\ref{tab:mechanisms} and plotted for select parameters in Figures~\ref{fig:RDP} and~\ref{fig:rr_laplace}.

\begin{table*}
\begin{center}
\begin{tabular}{p{.225\textwidth}p{.225\textwidth}p{.45\textwidth}}
\toprule
\textbf{Mechanism} & \textbf{Differential Privacy} & \textbf{\renyi\ Differential Privacy for $\alpha$}\\
\midrule
\multirow{2}{.225\textwidth}{Randomized Response} & \multirow{2}{.225\textwidth}{$\left|\log \frac{p}{1-p}\right|$}&$\alpha>1$: $\frac1{\alpha-1}\log \left(p^\alpha(1-p)^{1-\alpha}+(1-p)^\alpha p^{1-\alpha}\right)$\\
&&$\alpha=1$: $(2p-1)\log\frac{p}{1-p}$\\
\cmidrule{1-3}
\multirow{2}{.225\textwidth}{Laplace Mechanism} & \multirow{2}{.225\textwidth}{$1/\lambda$}
&$\alpha>1$: $\frac1{\alpha-1}\log\left\{\frac{\alpha}{2\alpha-1}\exp(\frac{\alpha-1}\lambda)+\frac{\alpha-1}{2\alpha-1}\exp(-\frac\alpha\lambda)\right\}$\\
&&$\alpha=1$: $1/\lambda+\exp(-1/\lambda)-1=.5/\lambda^2+O(1/\lambda^3)$\\
\cmidrule{1-3}
Gaussian Mechanism & $\infty$
&$\alpha/(2\sigma^2)$\\
\bottomrule\end{tabular}
\end{center}
\caption{Summary of RDP parameters for basic mechanisms.}\label{tab:mechanisms}
\end{table*}

\subsection{Randomized response}

Let $f$ be a predicate, i.e., $f\colon\D\mapsto\{0,1\}$. The Randomize Response mechanism for $f$ is defined as
\[
\RRM_p f(D) \triangleq \begin{cases} f(D)\quad &\textrm{with probability $p$}\\1-f(D)\quad &\textrm{with probability $1-p$}\end{cases}.
\]

The following statement can be verified by direct application of the definition of \renyi\ differential privacy:

\begin{proposition}\label{prop:randresponse} Randomized Response mechanism $\RRM_p(f)$ satisfies
\[
\left(\alpha,\frac1{\alpha-1}\log \left(p^\alpha(1-p)^{1-\alpha}+(1-p)^\alpha p^{1-\alpha}\right)\right)\textrm{-RDP}
\]
if $\alpha > 1$, and
\[ 
\left(\alpha,(2p-1)\log\frac{p}{1-p}\right)\textrm{-RDP}
\]
if $\alpha = 1$.
\end{proposition}


\subsection{Laplace noise}

Through the rest of this section we assume that $f\colon \D\mapsto\mathbb{R}$ is a function of sensitivity 1, i.e., for any two adjacent $D,D'\in\D$: $\left|f(D)-f(D')\right|\leq 1$.

Define the Laplace mechanism for $f$ of sensitivity 1 as
\[
\LaplaceM_\lambda f(D) = f(D) + \Lambda(0,\lambda),
\]
where $\Lambda(\mu,\lambda)$ is Laplace distribution with mean $\mu$ and scale $\lambda$, i.e., its probability density function is $\frac1{2\lambda}\exp(-|x-\mu|/\lambda)$.

To derive the RDP budget curve for the exponential mechanism we compute the \renyi\ divergence for Laplace distribution and its offset.

\begin{proposition}\label{prop:laplacian} For any $\alpha\geq 1$ and $\lambda >0$:
\begin{multline*}
 D_\alpha(\Lambda(0,\lambda)\|\Lambda(1,\lambda))=
 \frac1{\alpha-1}\log\left\{\frac{\alpha}{2\alpha-1}\exp\left(\frac{\alpha-1}{\lambda}\right)\right.\\
 \left.+\frac{\alpha-1}{2\alpha-1}\exp\left(\frac{-\alpha}{\lambda}\right)\right\}.
\end{multline*}
\end{proposition}
\begin{proof}
For continuous distributions $P$ and $Q$ defined over the real interval with densities $p$ and $q$
\[
D_\alpha(P\|Q)=\frac1{\alpha-1}\log \int_{-\infty}^{\infty} p(x)^{\alpha}q(x)^{1-\alpha} \dx.
\]

To compute the integral for $p(x)=\frac1{2\lambda}\exp(-|x|/\lambda)$ and $q(x)=\frac1{2\lambda}\exp(-|x-1|/\lambda)$, we evaluate it separately over the intervals $(-\infty,0]$, $[0,1]$ and $[1,+\infty]$.
\begin{align*}
\int_{-\infty}^{+\infty}& p(x)^{\alpha}q(x)^{1-\alpha} \dx =\\
&\phantom{+\>\>}\frac1{2\lambda}\int_{-\infty}^{0} \exp(\alpha x/\lambda + (1-\alpha)(x-1)/\lambda)\dx\\
&+\frac1{2\lambda}\int_0^1 \exp(-\alpha x/\lambda + (1-\alpha)(x-1)/\lambda)\dx\\
&+\frac1{2\lambda}\int_1^{+\infty} \exp(-\alpha x/\lambda - (1-\alpha)(x-1)/\lambda)\dx\\
=&\phantom{+\>\>}\frac12\exp((\alpha-1)/\lambda)\\
&+\frac1{2(2\alpha-1)}\left(\exp((\alpha-1)/\lambda)-\exp(-\alpha/\lambda)\right)\\
&+\frac12\exp(-\alpha/\lambda)\\
=&\frac{\alpha}{2\alpha-1}\exp((\alpha-1)/\lambda)+\frac{\alpha-1}{2\alpha-1}\exp(-\alpha/\lambda),
\end{align*}
from which the claim follows.
\end{proof}

Since the Laplace mechanism is additive, the \renyi\ divergence between $\LaplaceM_\lambda f(D)$ and $\LaplaceM_\lambda f(D')$ depends only on $\alpha$ and the distance $|f(D)-f(D')|$. Proposition~\ref{prop:laplacian} implies the following:

\begin{corollary}If real-valued function $f$ has sensitivity 1, then the Laplace mechanism $\LaplaceM_\lambda f$ satisfies $(\alpha,$\linebreak $\frac1{\alpha-1}\log\left\{\frac{\alpha}{2\alpha-1}\exp(\frac{\alpha-1}\lambda)+\frac{\alpha-1}{2\alpha-1}\exp(-\frac\alpha\lambda)\right\})$-RDP.
\end{corollary}

Predictably, 
\begin{multline*}
\lim_{\alpha\to\infty} D_\alpha(\Lambda(0,\lambda)\|\Lambda(1,\lambda))=
D_\infty(\Lambda(0,\lambda)\|\Lambda(1,\lambda))=\frac1\lambda.
\end{multline*}
This is, of course, consistent with the Laplace mechanism satisfying $1/\lambda$-differential privacy. The other extreme evaluates to the following expression $\lim_{\alpha\to1} D_\alpha(\Lambda(0,\lambda)\|\Lambda(1,\lambda))=1/\lambda+\exp(-1/\lambda)-1$, which is well approximated by $.5/\lambda^2$ for large~$\lambda$.

\subsection{Gaussian noise}

Assuming, as before, that $f$ is a real-valued function, the Gaussian mechanism for approximating $f$ is defined as 
\[
\GaussianM_\sigma f(D)=f(D)+N(0,\sigma^2),
\]
where $N(0,\sigma^2)$ is normally distributed random variable with standard deviation $\sigma^2$ and mean 0.

The following statement is a closed-form expression of the \renyi\ divergence between a Gaussian and its offset (for a more general version see~\cite{EH07-Renyi,Liese-Vajda}).
\begin{proposition} $D_\alpha(N(0,\sigma^2)\|N(\mu,\sigma^2))=\alpha\mu^2/(2\sigma^2).$
\end{proposition}
\begin{proof}
By direct computation we verify that
\begin{align*}
D_\alpha&(N(0,\sigma^2)\|N(\mu,\sigma^2))\\
&=\frac1{\alpha-1}\log\int_{-\infty}^{\infty}\frac1{\sigma\sqrt{2\pi}}\exp(-\alpha x^2/(2\sigma^2))\\
&\quad\quad\quad\quad\cdot \exp(-(1-\alpha) (x-\mu)^2/(2\sigma^2))\dx\\
&=\frac1{\alpha-1}\log\frac1{\sigma\sqrt{2\pi}}\int_{-\infty}^{\infty} \exp[(-x^2+\\
&\quad\quad\quad\quad 2(1-\alpha)\mu x-(1-\alpha)\mu^2)/(2\sigma^2)]\dx\\
&=\frac1{\alpha-1}\log\left\{\frac{\sigma\sqrt{2\pi}}{\sigma\sqrt{2\pi}}\exp\left[(\alpha^2-\alpha)\mu^2/(2\sigma^2)\right]\right\}\\
&=\alpha\mu^2/(2\sigma^2).
\end{align*}
\end{proof}

The following corollary is immediate:
\begin{corollary}If $f$ has sensitivity 1, then the Gaussian mechanism $\GaussianM_\sigma f$ satisfies $(\alpha, \alpha/(2\sigma^2))$-RDP.
\end{corollary}

Observe that the RDP budget curve for the Gaussian mechanism is particularly simple---a straight line (Figure~\ref{fig:RDP}). Recall that the (adaptive) composition of RDP mechanisms satisfies \renyi\ differential privacy with the budget curve that is the sum of the mechanisms' budget curves. It means that a composition of Gaussian mechanisms will behave, privacy-wise, ``like'' a Gaussian mechanism. Concretely, a composition of $n$ Gaussian mechanisms each with parameter $\sigma$ will have the RDP curve of a Gaussian mechanism with parameter $\sigma/\sqrt{n}$.

\begin{figure*}
	\centering
	\input{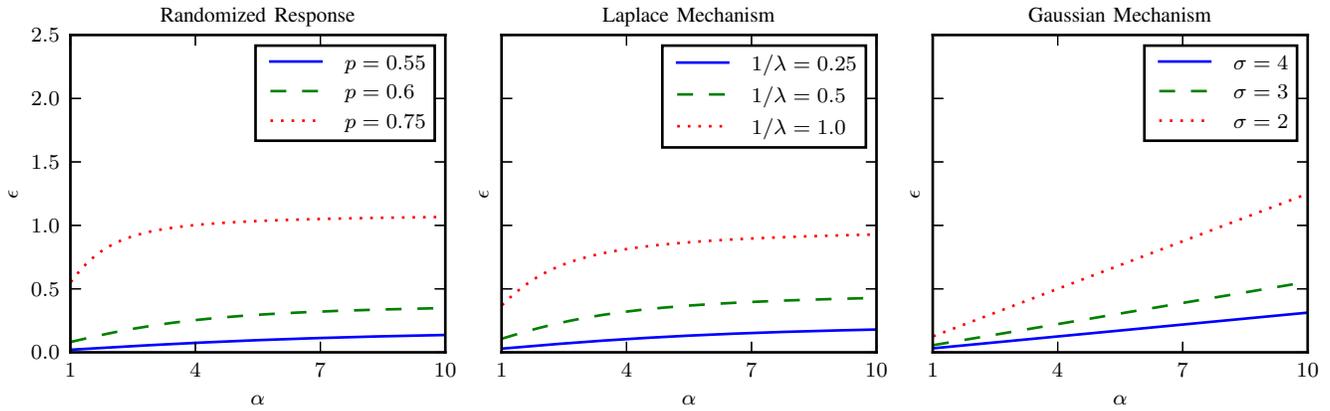}
	\caption{\alphaeps-\renyi\ differential privacy budget curve for three basic mechanisms with varying parameters. \label{fig:RDP}}
\end{figure*}

\begin{figure*}
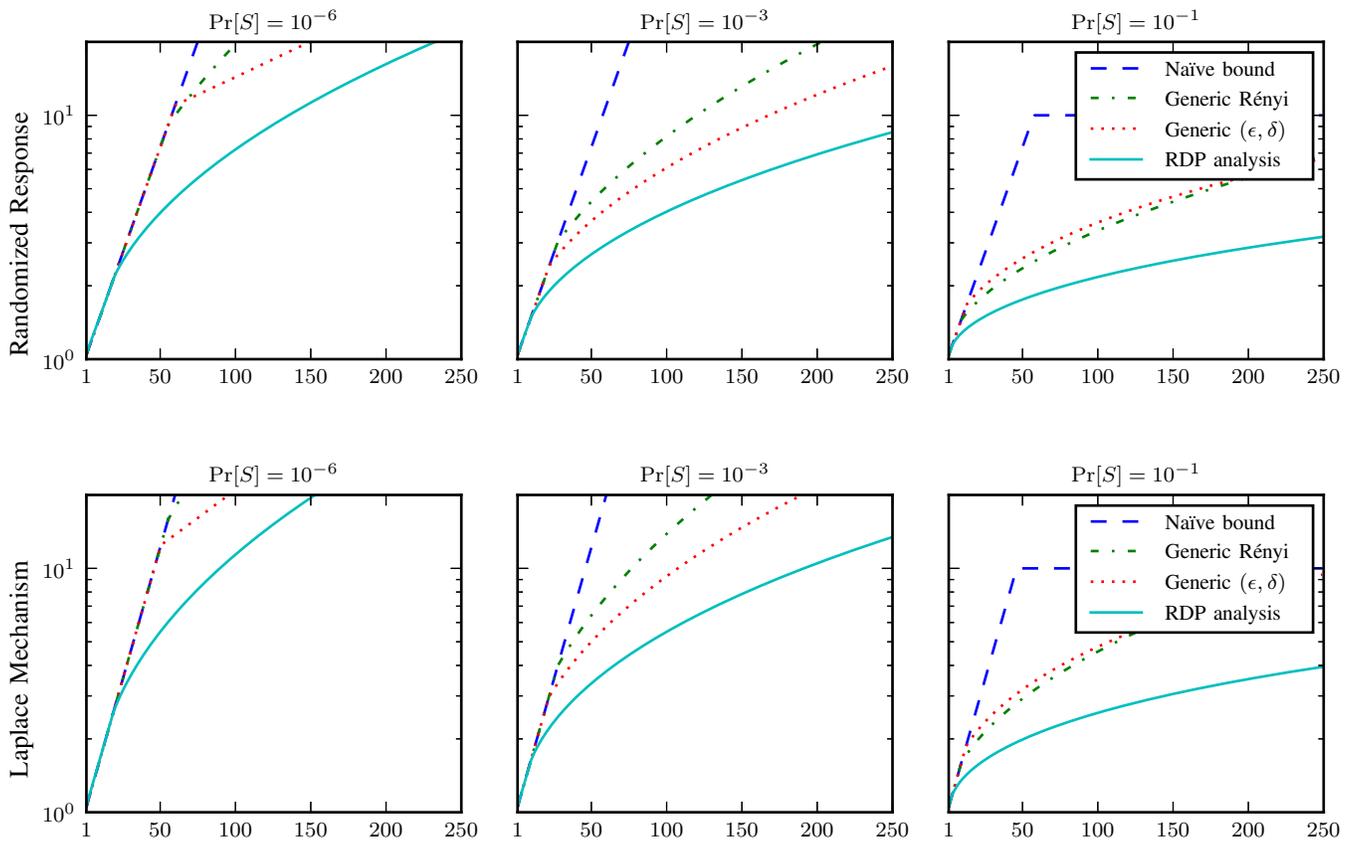

	\centering
	\input{Python/randresp.pgf}
	\input{Python/laplace.pgf}
	\caption{Various privacy guarantees of the randomized response with parameter $p=51\%$ (top row) and the Laplace mechanism with parameter $\lambda=20$ (bottom row) under self-composition. The $x$-axis is the number of compositions (1--250). The $y$-axis, in log scale, is the upper bound on the multiplicative increase in probability of event $S$, where $S$'s initial mass is either $10^{-6}$ (left), $10^{-3}$ (center), or $.1$ (right). The four plot lines are the ``na\"ive'' $n\eps$ bound (blue); optimal choice $\epsdelta$ in the standard advanced composition theorem (red); generic bound of Proposition~\ref{prop:advanced_comp} (blue); optimal choice of $\alpha$ in Proposition~\ref{prop:preservation} (cyan). \label{fig:rr_laplace}}
\end{figure*}

\subsection{Privacy of basic mechanisms under composition}

The ``bad outcomes'' interpretation of \renyi\ differential privacy ties the probabilities of seeing the same outcome under runs of the mechanism applied to adjacent inputs. The dependency of the upper bound on the increase in probability on its initial value is complex, especially compared to the standard differential privacy guarantee. The main advantage of this more involved analysis is that for most parameters the bound becomes tighter.

In this section we compare numerical bounds for several analyses of self-composed mechanisms (see Figure~\ref{fig:rr_laplace}), presented as three sets of graphs, where $\Pr[f(D)\in S]$ takes values $10^{-6}$, $10^{-3}$, and $10^{-1}$.

Each of the six graphs in Figure~\ref{fig:rr_laplace} (three probability values $\times$ \{randomized response, Laplace\}) plots bounds in logarithmic scale on the relative increase in probability of $S$ (i.e., $\Pr[f(D')\in S]/\Pr[f(D)\in S]$) offered by four analyses. The first, ``na\"ive'', bound follows from the basic composition theorem for differential privacy and, as expected, is very pessimistic for all but a handful of parameters. A tighter, advanced composition theorem~\cite{DRV10-boosting}, gives a choice of $\delta$, from which one computes $\eps'$ so that the $n$-fold composition satisfies $(\eps',\delta)$-differential privacy. The second curve plots the bound for the optimal (tightest) choice of $(\eps',\delta)$. Two other bounds come from \renyi\ differential privacy analysis: our generic advanced composition theorem (Proposition~\ref{prop:advanced_comp}) and the bound of Proposition~\ref{prop:preservation} for the optimal combination of $(\alpha,\eps)$ from the RDP curve of the composite mechanism.

Several observations are in order. The RDP-specific analysis for both mechanisms is tighter than all generic bounds whose only input is the mechanism's differential privacy parameter. On the other hand, our version of the advanced composition bound (Proposition~\ref{prop:advanced_comp}) is consistently outperformed by the standard $(\eps,\delta)$-form of the composition theorem, where $\delta$ is chosen \emph{optimally}. We elaborate on this distinction in the next section.

\begin{figure*}
	\centering
	\input{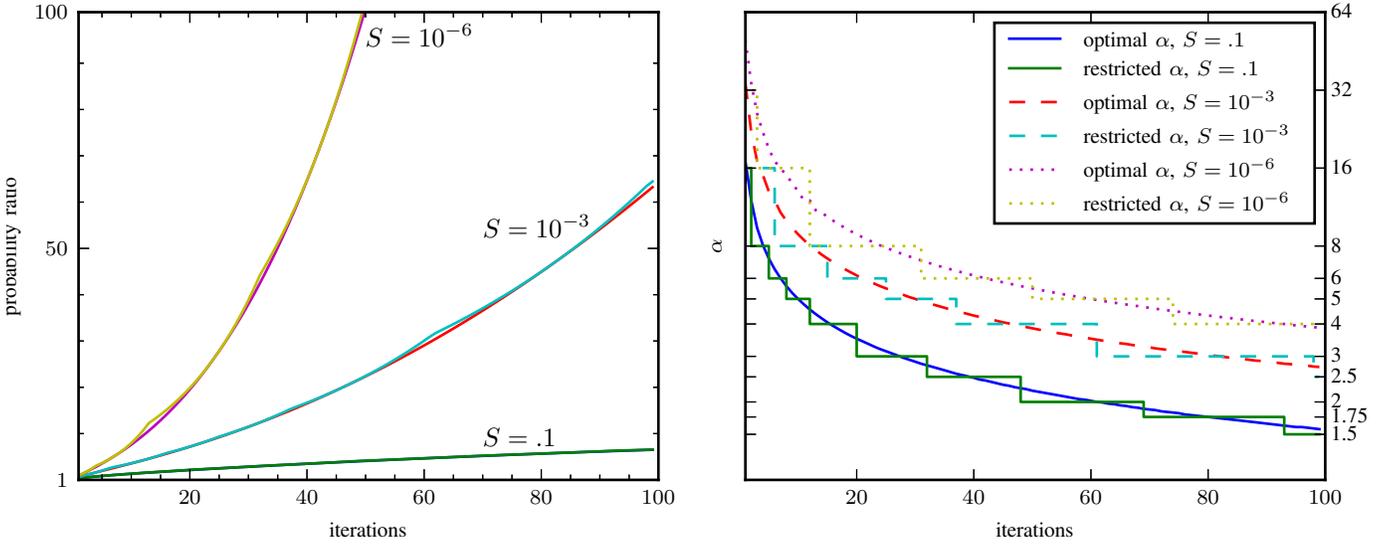}
	\caption{Left: Bounds on the ratio $\Pr[f(D')\in S]/\Pr[f(D)\in S]$ for $\Pr[f(D)\in S]\in\{.1, 10^{-3}, 10^{-6}\}$ for up to 100 iterations of a mixed mechanism (randomized response with $p=.52$, Laplace with $\lambda=20$ and Gaussian with $\sigma=10$). Each bound is computed twice: once for an optimal choice of $\alpha$ and once for $\alpha$ restricted to $\{1.5, 1.75, 2, 2.5, 3, 4, 5, 6, 8, 16, 32, 64,+\infty\}$. The curves for two choices of $\alpha$ are nearly identical. Right: corresponding values of $\alpha$ in log scale.\label{fig:alphas}}
\end{figure*}

\section{Discussion}\label{s:discussion}

\renyi\ differential privacy is a natural relaxation of the standard notion of differential privacy that preserves many of its essential properties. It can most directly be compared with \epsdelta-differential privacy, with which it shares several important characteristics. 

\leader{Probabilistic privacy guarantee.} The standard ``bad outcomes'' guarantee of \eps-differential privacy is independent of the probability of a bad outcome: it may increase only by a factor of $\exp(\eps)$. Its relaxation, \epsdelta-differential privacy, allows for an additional $\delta$ term, which allows for a complete privacy compromise with probability $\delta$.

In stark contrast, \renyi\ differential privacy even with very weak parameters never allows a total breach of privacy with no residual uncertainty. The following analysis quantifies this assurance.

Let $f$ be \alphaeps-RDP with $\alpha>1$. Recall that for any two adjacent inputs $D$ and $D'$, and an arbitrary prior $p$ the odds function $R(D,D')\sim p(D)/p(D')$ satisfies $\E\left[\left\{\frac{R_\textrm{post}(D,D')}{R_\textrm{prior}(D,D')}\right\}^{\alpha-1}\right]\leq\exp((\alpha-1)\eps)$. By Markov's inequality $\Pr[R_\textrm{post}(D,D')>\beta R_\textrm{prior}(D,D')]<e^\eps/\beta^{1/(\alpha-1)}$. For instance, if $\alpha=2$, the probability that the ratio between two posteriors increases by more than the $\beta$ factor drops off as $O(1/\beta)$.

\leader{Baseline-dependent guarantees.} The \renyi\ differential privacy bound gets weaker for less likely outcomes. For instance, if $f$ is a $(10.0, .1)$-RDP mechanism, an event of probability .5 under $f(D)$ can be as large as .586 and as small as .419 under $f(D')$. For smaller events the range is (in relative terms) wider. If the probability under $f(D)$ is .001, then $\Pr[f(D')\in S]\in [.00042, 0.00218]$. For $\Pr[f(D)\in S] = 10^{-6}$ the range is wider still: $\Pr[f(D')\in S]\in [.195\cdot10^{-6},4.36\cdot10^{-6}]$.

Contrasted with the pure \eps-differential privacy this type of guarantee is conceptually weaker and more onerous in application: in order to decide whether the increased risk is tolerable, one is required to estimate the baseline risk first.

However, in comparison with \epsdelta-DP the analysis via \renyi\ differential privacy is simpler and, especially for probabilities that are smaller than $\delta$, leads to stronger bounds. The reason is that \epsdelta-differential privacy often arises as a result of some analysis that implicitly comes with an \eps-$\delta$ tradeoff. Finding an optimal value of \epsdelta\ given the baseline risk may be non-trivial, especially in closed form. Contrast the following two, basically equivalent, statements of advanced composition theorems (Proposition~\ref{prop:advanced_comp} and its Corollary~\ref{cor:advanced_comp_epsdelta}):
\begin{quote}
Let $f\colon\D\mapsto\R$ be an adaptive composition of $n$ mechanisms all satisfying \eps-differential privacy for $\eps\leq 1$. Let $D$ and $D'$ be two adjacent inputs. Then for any $S\subset\R$, by Proposition~\ref{prop:advanced_comp}:
\begin{multline*}
\Pr[f(D')\in S] \leq
 \exp\left(2\eps\sqrt{n\log1/\Pr[f(D)\in S]}\right)\\
  \cdot \Pr[f(D)\in S].
\end{multline*}
or, by Corollary~\ref{cor:advanced_comp_epsdelta},
\begin{multline*}
\Pr[f(D')\in S] \leq
\exp\left(4\eps\sqrt{2 n \log 1 / \delta}\right)\\
\cdot \Pr[f(D)\in S]+\delta,
\end{multline*}
where $0<\eps,\delta<1$ such that $\log(1 / \delta) \geq \eps^2 n$.
\end{quote}

Given some value of baseline risk $\Pr[f(D)\in S]$, which formulation is easier to interpret? We argue that it is the former, since the \epsdelta\ form has a free parameter ($\delta$) that ought to be optimized in order to extract a tight bound that Proposition~\ref{prop:advanced_comp} gives directly.

The use of  \epsdelta\ bounds gets even more complex if we consider a composition of heterogeneous mechanisms. It brings us to the last point of comparison between  \epsdelta- and \renyi\ differential privacy measures.

\leader{Keeping track of accumulated privacy loss.} A finite privacy budget associated with an individual is an intuitive and appealing concept, to which \eps-differential privacy gives a rigorous mathematical expression. Cumulative loss of differential privacy over the cause of a mechanism run, a protocol, or one's lifetime can be tracked easily thanks to the additivity property of differential privacy. Unfortunately, doing so na\"ively likely exaggerates privacy loss, which grows sublinearly in the number of queries with all but negligible probability (via advanced composition theorems).

Critically, applying advanced composition theorems breaks the convenient abstraction of privacy as a non-negative real number. Instead, the guarantee comes in the \epsdelta\ form that effectively corresponds to a single point on an implicitly defined curve. Composition of multiple, heterogeneous mechanisms makes applying the composition rule optimally much more challenging, as one may choose various \epsdelta\ points to represent their privacy (in the analysis, not during the mechanisms' run time!). It begs the question of how to represent the privacy guarantee of a complex mechanism: distilling it to a single number throws away valuable information, while publishing the entire \epsdelta\ curve shifts the problem to the aggregation step. (See~Kairouz et al.~\cite{KOV15-composition} for an optimal bound on composition of homogeneous mechanisms and Murtagh and Vadhan~\cite{MV16-optimal} for hardness results and an approximation scheme for composition of mechanisms with heterogeneous privacy guarantees.)

\renyi\ differential privacy restores the concept of a privacy budget, thanks to its composition rule: RDP curves for composed mechanisms simply add up. Importantly, the $\alpha$'s of $(\alpha,\eps)$-\renyi\ differential privacy do not change. If RDP statements are reported for a common set of $\alpha$'s (which includes $+\infty$, to keep track of pure differential privacy), RDP of the aggregate is the sum of the reported vectors. Since the composition theorem of Proposition~\ref{prop:advanced_comp} takes as an input the mechanism's RDP curve, it means that the sublinear loss of privacy as a function of the number of queries will still hold.

For an example of this approach we tabulate the bound on privacy loss for an iterative mechanism consisting of three basic mechanisms: randomized response, Gaussian, and Laplace. Its RDP curve is given, in the closed form, by application of the basic composition rule to RDP curves of the underlying mechanisms (Table~\ref{tab:mechanisms}). The privacy guarantee is presented in Figure~\ref{fig:alphas} for three values of the baseline risk: .1, .001, and $10^{-6}$. For each set of parameters two curves are plotted: one for an optimal value of $\alpha$ from $(1,+\infty]$, the other for an optimal $\alpha$ restricted to the set of 13 values $\{1.5, 1.75, 2, 2.5, 3, 4, 5, 6, 8, 16, 32, 64,+\infty\}$. The two curves are nearly identical, which illustrates our thesis that reporting RDP curves for a restricted set of $\alpha$'s preserves tightness of privacy analysis.

\section{Conclusions and Open Questions}\label{s:conclusions}

We put forth the proposition that \renyi\ divergence yields useful insight into analysis of differentially private mechanisms. Among our findings
\begin{itemize}
\item \renyi\ differential privacy (RDP) is a natural generalization of pure differential privacy. 

\item RDP shares, with some adaptations, many properties that make differential privacy a useful and versatile tool. 

\item RDP analysis of Gaussian noise is particularly simple.

\item A composition theorem can be proved based solely on the properties of RDP, which implies that RDP packs sufficient information about a composite mechanism as to enable its analysis without consideration of its components.

\item Furthermore, an RDP curve may be sampled in just a few points to provide useful guarantees for a wide range of parameters. If these points are chosen consistently across multiple mechanisms, this information can be used to estimate aggregate privacy loss.
\end{itemize}

Naturally, multiple questions remain open. Among those
\begin{itemize}


\item As Lemma~\ref{lemma:D_alpha_bound} demonstrates, the RDP curve of a differentially private mechanism is severely constrained. Making fuller use of these constraints is a promising direction, in particular towards formal bounds on tightness of RDP guarantees from select $\alpha$ values.

\item Proposition~\ref{prop:preservation} (probability preservation) is not tight when $D_\alpha(P\|Q)\to 0$. We expect that $P(A)\to Q(A)$ but the bound does not improve beyond $P(A)^{(\alpha-1)/\alpha}$.
\end{itemize}

\section*{Acknowledgments}

We would like to thank Cynthia Dwork, Kunal Talwar, Salil Vadhan, and Li Zhang for numerous fruitful discussions, the CSF reviewers, Nicolas Papernot and Damien Desfontaines for their helpful comments, and Mark Bun and Thomas Steinke for sharing a draft of~\cite{BS16-zCDP}.

\bibliographystyle{IEEEtran}
\bibliography{renyi}

\appendix

\section{Basic Properties of the \renyi\ Divergence}

For comprehensive exposition of properties of the \renyi\ divergence we refer to two recent papers~\cite{EH07-Renyi,Shayevitz11-Renyi}. Here we recall and re-prove several facts useful for our analysis.

\begin{proposition}[Non-negativity]\label{prop:non-negativity} For $1\leq \alpha$ and arbitrary distributions $P,Q$
	\[
	D_\alpha(P\|Q)\geq 0.
	\]
\end{proposition}
\begin{proof} Assume that $\alpha>1$. Define $\phi(x)\triangleq x^{1-\alpha}$ and $g(x) \triangleq Q(x)/P(x)$. Then
	\begin{align*}
	D_\alpha(P\|Q)&=\frac1{\alpha-1}\log \E_P[\phi(g(x))]\\
	&\geq \frac1{\alpha-1}\log \phi(\E_P[g(x)])\\
	&=0
	\end{align*}
	by the Jensen inequality applied to the convex function $\phi$. The case of $\alpha=1$ follows by letting $\phi$ to be $\log 1/x$. 
\end{proof}

\begin{proposition}[Monotonicity]\label{prop:monotonicity} For $1\leq \alpha<\beta$ and arbitrary $P,Q$
	\[
	D_\alpha(P\|Q)\leq D_\beta(P\|Q).
	\]
\end{proposition}
\begin{proof}[Proof (following~\cite{EH07-Renyi})] Assume that $\alpha>1$. Observe that the function $x\mapsto x^{\frac{\alpha-1}{\beta-1}}$ is concave. By Jensen's inequality
	\begin{align*}
	D_\alpha(P\|Q)&=\frac{1}{\alpha-1}\log \E_P \left(\frac{P(x)}{Q(x)}\right)^{\alpha-1}\\
	&=\frac1{\alpha-1}\log \E_P\left(\frac{P(x)}{Q(x)}\right)^{(\beta-1){\frac{\alpha-1}{\beta-1}}}\\
	&\leq \frac1{\alpha-1}\log\left\{ \E_P\left(\frac{P(x)}{Q(x)}\right)^{\beta-1}\right\}^{\frac{\alpha-1}{\beta-1}}\\
	&=D_\beta(P\|Q).\\
	\end{align*}
	The case of $\alpha=1$ follows by continuity.
\end{proof}

The following proposition appears in Langlois et al.~\cite{LSS14-GGHLite}, generalizing Lyubashevsky et al.~\cite{LPR13}.

\begin{proposition}[Probability preservation~\cite{LSS14-GGHLite}]\label{prop:preservation} Let $\alpha>1$, $P$ and $Q$ be two distributions defined over $\R$ with identical support, $A\subset \R$ be an arbitrary event. Then 
	\[
	P(A)\leq \left(\exp[D_\alpha(P\|Q)]\cdot Q(A)\right)^{(\alpha-1)/\alpha}.
	\]
\end{proposition}
\begin{proof} The result follows by application of H\"older's inequality, which states that for real-valued functions $f$ and $g$, and real $p,q>1$, such that $1/p+1/q=1$,
	\[\|fg\|_1\leq \|f\|_p\|g\|_q.\]
	By setting $p\triangleq\alpha$, $q\triangleq\alpha/(\alpha-1)$, $f(x)\triangleq P(x)/Q(x)^{1/q}$, $g(x)\triangleq Q(x)^{1/q}$, and applying H\"older's, we have
	\begin{align*}\int_A\!\! P(x)\dx&\leq \left(\int_A\!\! P(x)^{\alpha}Q(x)^{1-\alpha}\dx\right)^\frac1\alpha\!\!\left(\int_A\!\! Q(x)\dx\right)^{\frac{\alpha-1}{\alpha}}\\
	&\leq \exp[D_\alpha(P\|Q)]^{(\alpha-1)/\alpha}Q(A)^{(\alpha-1)/\alpha},
	\end{align*}
	completing the proof.
\end{proof}

The most salient feature of the bound is its (often non-monotone) dependency on $\alpha$: as $\alpha$ approaches~1, $D_\alpha(P\|Q)$ shrinks (by monotonicity of the \renyi\ divergence) but the power to which it is raised goes to~0, pushing the result in the opposite direction. Several our proofs proceed by finding the optimal, or approximately optimal, $\alpha$ minimizing the bound.

The \renyi\ divergence is not a metric: it is not symmetric and it does not satisfy the triangle inequality. A weaker variant of the triangle inequality tying together the \renyi\ divergence of different orders does hold. Its general version is presented below.

\begin{proposition}[Weak triangle inequality]\label{prop:triangle}Let $P,Q,R$ be distributions on $\R$. Then for $\alpha > 1$ and for any $p, q>1$ satisfying $1/p+1/q=1$ it holds that
	\[
	D_\alpha(P\|Q)\leq \frac{\alpha-1/p}{\alpha-1} D_{p\alpha}(P\|R) + D_{q(\alpha-1/p)}(R\|Q).
	\]
\end{proposition}
\begin{proof}
	By H\"older's inequality we have:
	\begin{align*}
	\exp&[(\alpha-1)D_\alpha(P\|Q)]\\
	&=\int_\R P(x)^\alpha Q(x)^{1-\alpha}\dx\\
	&=\int_\R \frac{P(x)^\alpha}{R(x)^{\alpha-1/p}}\frac{R(x)^{\alpha-1/p}}{Q(x)^{\alpha-1}}\dx\\
	&\leq \left\{\int_\R \frac{P(x)^{p\alpha}}{R(x)^{p\alpha-1}}\dx\right\}^{1/p}
	\left\{\int_\R \frac{R(x)^{q\alpha-q/p}}{Q(x)^{q\alpha-q}}\dx\right\}^{1/q}\\
	&=\exp[(\alpha-1/p) D_{p\alpha}(P\|R)]\cdot\\
	&\quad \exp[(\alpha-1)D_{q\alpha-q/p}(R\|Q)].
	\end{align*}
	By taking the logarithm and dividing both sides by $\alpha-1$ we establish the claim.
\end{proof}

Several important special cases of the weak triangle inequality can be obtained by fixing parameters $p$ and $q$ (compare it with~\cite[Lemma 12]{MMR09-Renyi} and~\cite[Lemma 4.1]{LSS14-GGHLite}):
\begin{corollary}\label{col:triangle}For $P, Q, R$ with common support we have
	\begin{enumerate}
		\item $D_\alpha(P\|Q)\leq \frac{\alpha-1/2}{\alpha-1}D_{2\alpha}(P\|R)+D_{2\alpha-1}(R\|Q).$
		\item $D_\alpha(P\|Q)\leq \frac{\alpha}{\alpha-1}D_{\infty}(P\|R)+D_{\alpha}(R\|Q).$
		\item $D_\alpha(P\|Q)\leq D_{\alpha}(P\|R)+D_{\infty}(R\|Q).$
		\item $D_\alpha(P\|Q)\leq \frac{\alpha-\alpha/\beta}{\alpha-1}D_\beta(P\|R)+D_\beta(R\|Q)$,
		for some explicit $\beta=2\alpha-.5 + O(1/\alpha)$.
	\end{enumerate}
\end{corollary}
\begin{proof}
	All claims follow from the weak triangle inequality (Proposition~\ref{prop:triangle}) where $p$ and $q$ are chosen, respectively, as
	\begin{enumerate}
		\item $p=q=2$.
		\item $p\to\infty$ and $q\triangleq p/(p-1)\to 1$.
		\item $q\to\infty$ and $p\triangleq q/(q-1)\to 1$.
		\item such that $\alpha p=\alpha q-1$ and $1/p+1/q=1$.
	\end{enumerate}
	In the last case $\beta\triangleq p\alpha=2\alpha-.5 + O(1/\alpha)$. 
\end{proof}

\end{document}